\documentclass[12pt]{article}
\usepackage{amsmath,amsthm,amssymb,amscd, latexsym}
\usepackage{bussproofs}
\usepackage{synttree}
\usepackage{hyperref}

\newtheorem{theorem}{Theorem}[section]
\newtheorem{proposition}[theorem]{Proposition}
\newtheorem{lemma}[theorem]{Lemma}
\newtheorem{definition}[theorem]{Definition}

\theoremstyle{remark}

\newtheorem{remark}[theorem]{Remark}
\newtheorem{note}[theorem]{Note}

\newtheorem{problem}[theorem]{Problem}

\newcommand{\F}{\mathbb{F}}
\newcommand{\N}{\mathbb{N}}

\def\[#1]{\hbox{$ [\kern -.4em [\, {#1}\, ]\kern -.4em]$}}
\newcommand{\means}[1]{\hbox{$ [\kern -.4em [\, {#1}\, ]\kern -.4em]$}}

\newcommand{\corner}[1]{{\ulcorner {#1} \urcorner}}
 
\newcommand{\beg}{\text{{\tt begin}}}
\newcommand{\fin}{\text{{\tt end}}}
\newcommand{\while}{\text{{\tt while}}}
\newcommand{\faire}{\text{{\tt do}}}
\newcommand{\ifa}{\text{ {\tt if} }}
\newcommand{\then}{\text{ {\tt then} }}
\newcommand{\elsea}{\text{ {\tt else} }}

\newcommand{\undef}{\text{{\tt undef}}}
\newcommand{\sta}{\text{sta}}
\newcommand{\dyn}{\text{dyn}}
\newcommand{\skipa}{\tt{Skip}}
\newcommand{\halt}{\tt{Halt}}
\newcommand{\fail}{\tt{Fail}}
\newcommand{\upd}{\text{{\tt Update}}}
\newcommand{\act}{\text{ {\tt Active} }}

\newcommand{\pad}{{\textit{pad}}}

\newcommand{\redu}{\rightarrow}

\newcommand{\Kmin}{K_{\textit{min}}}
\newcommand{\Lmin}{L_{\textit{min}}}

\newcommand{\redustar}{\twoheadrightarrow}

\newcommand{\betared}{\redu_\beta}
\newcommand{\betaredstar}{\redustar_\beta}

\newcommand{\dotsp}{\  .\ }

\newcommand{\Bool}{{\tt Bool}}
\newcommand{\true}{{\tt True}}
\newcommand{\false}{{\tt False}}

\newcommand{\slfalse}{\text{\sl{False}}}
\newcommand{\Case}{\text{\sl{Case}}}

\newcommand{\rem}{\text{{\tt rem}}}
\newcommand{\slrem}{\text{\sl{rem }}}

\newcommand{\slzero}{\text{\sl{Zero}}}

\newcommand{\slsuc}{\text{\sl{Succ}}}
\newcommand{\suc}{\text{\tt Succ}}
\newcommand{\slpred}{\text{\sl{Pred}}}

\begin{document}

\title{ASMs and Operational Algorithmic Completeness of Lambda Calculus}

\author{
Marie Ferbus-Zanda
\footnote{LIAFA, CNRS \& Universit\'e Paris Diderot - Paris 7,
Case 7014 \ 75205 Paris Cedex 13}\\
\newcounter{fnnumber}
\setcounter{fnnumber}{\value{footnote}}
{\footnotesize\tt
{\tt ferbus@liafa.jussieu.fr}}
\and
Serge Grigorieff
{\footnotemark[\value{fnnumber}]}\\
{\footnotesize\tt
http://www.liafa.jussieu.fr/$\sim$seg}
\\
{\footnotesize\tt
{\tt seg@liafa.jussieu.fr}}
}

\maketitle
\tableofcontents

\begin{abstract}
\noindent
We show that lambda calculus is a computation model
which can step by step simulate any sequential deterministic
algorithm for any computable function over integers or words
or any datatype.
More formally, given an algorithm above a family
of computable functions
(taken as primitive tools, i.e., kind of oracle functions for
the algorithm), 
for every constant K big enough,
each computation step of the algorithm can be simulated
by exactly K successive reductions in a natural extension
of lambda calculus with constants for functions in the
above considered family.
\\
The proof is based on a fixed point technique in lambda
calculus and on Gurevich sequential Thesis which allows
to identify sequential deterministic algorithms with
Abstract State Machines.
\\
This extends to algorithms for partial computable functions
in such a way that finite computations ending with exceptions
are associated to finite reductions leading to terms with a
particular very simple feature.
\end{abstract}
{\bf keywords.}
ASM, Lambda calculus, Theory of algorithms, Operational semantics

\section{Introduction}\label{s:intro}
%
%
\subsection{Operational versus Denotational Completeness}
\label{ss:operational}
%
Since the pioneering work of Church and Kleene, going back to 1935,
many computation models have been shown to compute the same
class of functions, namely, using Turing Thesis, the class of all
computable functions.
Such classes are said to be {\em Turing complete} or {\em denotationally
algorithmically complete}.

This is a result about crude input/output behaviour.
What about the ways to go from the input to the output,
i.e., the executions of algorithms in each of these computation models?
Do they constitute the same class?
Is there a Thesis for algorithms analog to Turing Thesis for computable
functions?

As can be expected, denotational completeness
does not imply operational completeness.
Clearly, the operational power of machines using massive parallelism cannot be matched by sequential machines.
For instance, on networks of cellular automata,
integer multiplication can be done in real time
(cf. Atrubin, 1962 \cite{atrubin},
see also Knuth, \cite{knuth} p.394-399),
whereas on Turing machines,
an $\Omega(n/\log n)$  time lower bound is known.
Keeping within sequential computation models,
multitape Turing machines have greater operational power
than one-tape Turing machines.
Again, this is shown using a complexity argument:
palindromes recognition can be done in linear
time on two-tapes Turing machines, whereas it requires
computation time $O(n^2)$ on one-tape Turing machines
(Hennie, 1965 \cite{hennie65}, see also \cite{multidim03,paul79}).

Though resource complexity theory may disprove operational
algorithmic completeness,
there was no formalization of a notion of operational completeness
since the notion of algorithm itself had no formal mathematical
modelization.
Tackled by Kolmogorov in the 50's \cite{kolmo58},
the question for {\em sequential algorithms} has been answered
by Gurevich
in the 80's \cite{gurevich84,gurevich85,gurevich91}
(see \cite{borger02} for a comprehensive survey of the question),
with their formalization as
``{\em evolving algebras}"
(now called {\em ``abstract state machines" or ASMs})
which has lead to {\em Gurevich's sequential Thesis}.
\medskip

Essentially, an ASM can be viewed as a first order
multi-sorted structure
and a program which modifies some of its predicates
and functions (called dynamic items).
Such dynamic items capture the moving environment of
a procedural program.
The run of an ASM is the sequence of structures
--~also called states~--
obtained by iterated application of the program.
The program itself includes
two usual ingredients of procedural languages, namely
affectation and the conditional ``if\ldots then\ldots else\ldots",
plus a notion of parallel block of instructions.
This last notion is a key idea which is somehow 
a programming counterpart to the mathematical notion
of system of equations.

Gurevich's sequential Thesis
\cite{gurevich85,gurevichSeqThesis99,gurevichSeqThesis00}
asserts that ASMs capture the notion of sequential algorithm.
Admitting this Thesis,
the question of operational completeness for a sequential
procedural computation model
is now the comparison of its operational power with that of ASMs.
%
%
\subsection{Lambda Calculus and Operational Completeness}
\label{ss:intromain}
%
In this paper we consider lambda calculus, a subject created by
Church and Kleene in the 30's,
which enjoys a very rich mathematical theory.
It may seem a priori strange to look for operational completeness
with such a computation model so close to an assembly language
(cf. Krivine's papers since 1994, e.g., \cite{krivine07}).
It turns out that, looking at reductions by groups
(with an appropriate but constant length),
and allowing one step reduction of primitive operations,
lambda calculus simulates ASMs in a very tight way.
Formally, our translation of ASMs in lambda calculus is as follows.
Given an ASM, we prove that,
for every integer $K$ big enough
(the least such $K$ depending on the ASM),
there exists a lambda term $\theta$ with the following property.
Let $a^t_1,\ldots,a^t_p$ be the values
(coded as lambda terms)
of all dynamic items of the ASM at step $t$, if the run
does not stop at step $t$ then
$$
\theta a^t_1\ldots a^t_p\quad
\overbrace{\redu \quad\cdots \quad\redu}^{K\text{ reductions}}
\quad\theta a^{t+1}_1\ldots a^{t+1}_p\ .
$$
If the run stops at step $t$ then the left term reduces to
a term in normal form which gives the list of outputs if they
are defined.
Thus, representing the state of the ASM at time $t$ by the term
$\theta a^t_1\ldots a^t_p$, a group of $K$ successive
reductions gives the state at time $t+1$.
In other words, $K$ reductions faithfully simulate one step of the ASM run.
Moreover, this group of reductions is that obtained by the
{\em leftmost redex reduction strategy},
hence it is a deterministic process.
Thus, {\em lambda calculus is operationally complete
for deterministic sequential computation}.

Let us just mention that adding to lambda calculus
one step reduction of primitive operations is not an unfair trick.
Every algorithm has to be ``above" some basic operations
which are kind of oracles:
the algorithm decomposes the computation in elementary steps
which are considered as atomic steps though they obviously
themselves require some work.
In fact, such basic operations can be quite complex:
when dealing with integer matrix product
(as in Strassen's algorithm in time $O(n^{\log 7})$),
one considers integer addition and multiplication as basic...
Building algorithms on such basic operations is indeed
what ASMs do with the so-called static items,
cf. \S\ref{ss:approach}, Point 2.

The proof of our results uses Curry's fixed point technique in lambda calculus
plus some padding arguments.

%
%
\subsection{Road Map}
\label{ss:road}
%
This paper deals with two subjects which have so far
not been much related:
ASMs and lambda calculus.
To make the paper readable to both ASM and lambda calculus
communities, the next two sections recall all needed prerequisites
in these two domains
(\emph{so that most readers may skip one of these two sections}).

What is needed about ASMs is essentially their definition,
but it cannot be given without a lot of preliminary notions
and intuitions.
Our presentation of ASMs in \S\ref{s:ASM} differs in inessential
ways from Gurevich's one
(cf.  \cite{gurevich91,%
gurevichDraft97,%
gurevichSeqThesis00,%
dershowitzgurevich08}).
Crucial in the subject (and for this paper) is Gurevich's sequential
Thesis that we state in \S\ref{ss:Thesis}.
We rely on the literature for the many arguments supporting this Thesis.

\S\ref{s:lambda} recalls the basics of lambda calculus,
including the representation of lists and integers and Curry fixed
point combinator.

The first main theorem in \S\ref{ss:main0} deals with
the simulation in lambda calculus of sequential algorithms
associated to ASMs in which all dynamic symbols are constant ones
(we call them type $0$ ASMs).
The second main theorem in \S\ref{ss:main1} deals with
the general case.
%
%
\section{ASMs}\label{s:ASM}
%
%
\subsection{The Why and How of ASMs on a Simple Example}
\label{ss:euclid}
%
\paragraph{Euclid's Algorithm}
Consider Euclid's algorithm to compute the greatest common divisor
(gcd) of two natural numbers.
It turns out that such a simple algorithm already allows to pinpoint
an operational incompleteness in usual programming languages.
Denoting by  $\rem(u,v)$ the remainder of $u$ modulo $v$,
this algorithm can be described as follows\footnote{Sometimes,
one starts with a conditional swap: if $a<b$ then $a,b$ are exchanged.
But this is done in the first round of the while loop.}
$$
\textit{\begin{tabular}{|l|}
\hline
Given data: two natural numbers $a,b$\\
While $b\neq0$ replace the pair $(a,b)$ by $(b,\rem(a,b))$\\
When $b=0$ halt: $a$ is the wanted gcd
\\\hline
\end{tabular}}\
$$
Observe that the the pair replacement in the above while loop
involves some elementary parallelism which is the algorithmic
counterpart to co-arity,
i.e., the consideration of functions with range in multidimensional spaces
such as the $\N^2\to\N^2$ function $(x,y)\mapsto(y,\rem(x,y))$.
%
\paragraph{Euclid's Algorithm in Pascal}
In usual programming languages, the above simultaneous replacement
is impossible: affectations are not done in parallel but sequentially.
For instance, {\em no Pascal program implements it as it is},
one can only get a distorted version with an extra algorithmic
contents involving a new variable $z$, cf. Figure \ref{fig:euclid}.
\begin{figure}\center
$\begin{array}{|l|}
\hline
	\begin{array}{l}
	\hline
	\text{\em Euclid's algorithm in Pascal}\\
	\hline
	\end{array}
\\\hline
	\begin{array}{l}
	\while\ b>0\ \faire\ \beg\\
	\phantom{\while\ b>0\ \faire\ \ } z:=a;\\
	\phantom{\while\ b>0\ \faire\ \ } a:=b;\\
	\phantom{\while\ b>0\ \faire\ \ } b:=\slrem(z,b);\\
	\phantom{\while\ b>0\ \faire\ \ } \fin;
	\end{array}\\
\ gcd:=a.
\\\hline
\end{array}
\qquad
\begin{array}{|c|}
\hline
         \begin{array}{l}
	\hline
	\text{\em Euclid's algorithm in ASM}\\
	\hline
	\end{array}
\\\hline\\
\begin{array}{l}
         \ifa 0<b \then \left|\begin{array}{l}
                                        a:=b\\
                                        b:=\slrem(a,b)
                                        \end{array}\right.\\
	\end{array}
\\\\\hline
\end{array}$
\\
(In both programs, $a,b$ are inputs
and $a$ is the output)
\caption{Pascal and ASM programs for Euclid's algorithm}
\label{fig:euclid}
\end{figure}
%
\paragraph{An ASM for Euclid's Algorithm}
Euclid's algorithm has a faithful formalization using an ASM.
The vertical bar on the left in the ASM program
(cf. Figure \ref{fig:euclid})
tells that the two updates are done simultaneously and independently.
Initialization gives symbols $a,b$ the integer values of which
we want to compute the gcd.
The semantical part of the ASM involves the set $\N$ of integers
to interpret all symbols.
Symbols $0,<,=,\slrem$ have fixed interpretations in integers
which are the expected ones.
Symbols $a,b$ have varying interpretations in the integers.
The sequence of values taken by $a,b$ constitutes the run of the ASM.
\\
When the instruction gets void (i.e., when $b$ is null)
the run stops and the value of the symbol $a$ is considered to be
the output.
%
\subsection{Gurevich Sequential Thesis}\label{ss:Thesis}
%
Yuri Gurevich has gathered as three Sequential Postulates
(cf. \cite{gurevichSeqThesis00,dershowitzgurevich08})
some key features of deterministic sequential algorithms for
partial computable functions (or type 1 functionals).
\begin{enumerate}
\item[I] {\em(Sequential time).}
An algorithm is a deterministic state-transition system.
Its transitions are partial functions.\\
Non deterministic transitions and
even nonprocedural input/output specifications
are thereby excluded from consideration.
\item[II] {\em(Abstract states).}
States are multitructures\footnote{In ASM theory,
an ASM is, in fact, a multialgebra (cf. point 1 of Remark \S\ref{rk:depart}).},
sharing the same fixed, finite vocabulary.
States and initial states are closed under isomorphism.
Transitions preserve the domain, and transitions and isomorphisms
commute.
\item[III] {\em(Bounded exploration).}
Transitions are determined by a fixed finite
``glossary" of ``critical" terms.
That is, there exists some finite set of (variable-free) terms
over the vocabulary of the states such that states that agree on
the values of these glossary terms also agree on all
next-step state changes.
\end{enumerate}
Gurevich, 2000 \cite{gurevichSeqThesis00},
stated an operational counterpart to Church's Thesis
:
{\bf Thesis.}[Gurevich's sequential Thesis]\label{Thesis}
{\emph Every sequential algorithm satisfies the Sequential Postulates I-III.}
%
%
\subsection{The ASM Modelization Approach}\label{ss:approach}
%
Gurevich's postulates lead to the following modelization approach
(we depart in non essential ways from \cite{dershowitzgurevich08},
see Remark \ref{rk:depart}).
\begin{enumerate}
\item
{\em The base sets.}
Find out the underlying families of objects involved in the given
algorithm, i.e., objects which can be values for inputs, outputs
or environmental parameters used during the execution of
the algorithm.
These families constitute the base sets of the ASM.
In Euclid's algorithm, a natural base set is the set $\N$
of natural integers. 
\item
{\em Static items.}
Find out
which particular fixed objects in the base sets are considered
and which functions and predicates over/between the base sets
are viewed as atomic in the algorithm,
i.e., are not given any modus operandi.
Such objects, functions and predicates are called the primitive
or static items of the ASM.
They do not change value through transitions.
In Euclid's algorithm, static items are
the integer $0$, the $\slrem$ function and the $<$ predicate.
\item
{\em Dynamic items.}
Find out the diverse objects, functions and predicates over
the base sets of the ASM which vary through transitions.
Such objects, functions and predicates are called the dynamic
items of the ASM.
In Euclid's algorithm, these are $a,b$.
\item\label{item:states}
{\em States: from a multi-sorted partial structure to a multi-sorted partial algebra.}
Collecting all the above objects, functions and predicates
leads to a first-order multi-sorted structure of some logical
typed language: any function goes from some product of sorts into some sort,
any predicate is a relation over some sorts.
However, there is a difference with the usual logical notion of
multi-sorted structure: predicates and functions may be partial.
A feature which is quite natural for any theory of computability,
a fortiori for any theory of algorithms.
\\
To such a multi-sorted structure one can associate a multi-sorted
algebra as follows.
First, if not already there, add a sort for Booleans.
Then replace predicates by their characteristic functions
In this way, we get a multi-sorted structure with partial functions only,
i.e. a multialgebra.
\item
{\em Programs.}
Finally, the execution of the algorithm can be viewed as a sequence
of states.
Going from one state to the next one amounts to applying
to the state a particular program --~called the ASM program~--
which modifies the interpretations of the sole dynamic symbols
(but the universe itself and the interpretations of
the static items remain unchanged).
Thus, the execution of the algorithm appears as an iterated application
of the ASM program. It is called the run of the ASM.
\\
Using the three above postulates, Gurevich
\cite{gurevichSeqThesis99,gurevichSeqThesis00}
proves that quite elementary instructions
--~namely blocks of parallel conditional updates~--
suffice to get ASM programs able to simulate step by step any
deterministic procedural algorithm.
\item
{\em Inputs, initialization map and initial state.}
Inputs correspond to the values of some distinguished static
symbols in the initial state, i.e., we consider that all inputs are given
when the algorithm starts
(though questionable in general, this assumption is reasonable
when dealing with algorithms to compute a function).
All input symbols have arity zero for algorithms computing functions.
Input symbols with non zero arity are used when dealing with
algorithms for type 1 functionals.
\\
The initialization map associates to each dynamic symbol a term
built up with static symbols.
In an initial state, the value of a dynamic symbol is required
to be that of the associated term given by the initialization map.
\item
{\em Final states and outputs.}
There may be several outputs, for instance
if the algorithm computes a function $\N^k\to\N^\ell$ with $\ell\geq2$.
\\
A state is final when, applying the ASM program to that state,
\begin{enumerate}
\item
either the $\halt$ instruction is executed {\em (Explicit halting)},
\item
or no update is made
(i.e. all conditions in conditional blocks of updates get value $\slfalse$)
{\em (Implicit halting)} .
\end{enumerate}
In that case, the run stops and the outputs correspond to the values of some distinguished dynamic symbols.
For algorithms computing functions, all output symbols are constants
(i.e. function symbols with arity zero).
\item
{\em Exceptions.}
There may be a finite run of the ASM ending in
a non final state. This corresponds to exceptions in programming
(for instance a division by $0$)
and there is no output in such cases. This happens when
\begin{enumerate}
\item
either the $\fail $ instruction is executed {\em (Explicit failing)},
\item
or there is a clash between two updates which are to be done
simultaneously {\em (Implicit failing)}.
\end{enumerate}
\end{enumerate}
\begin{remark}\label{rk:depart}
Let us describe how our presentation of ASMs (slightly)
departs from \cite{dershowitzgurevich08}.
\\
1. We stick to what Gurevich says
in \S.2.1 of \cite{gurevichlipari} (Lipari Guide, 1993):
{\em ``Actually, we are interested in multi-sorted structures
with partial operations"}.
Thus, we do not regroup sorts into a single universe
and do not extend functions with the {\em undef} element.
\\
2. We add the notion of initialization map which brings
a syntactical counterpart to the semantical notion of initial state.
It also rules out any question about the status of initial values
of dynamic items which would not be inputs.
\\
3. We add explicit acceptance and rejection as specific
instructions in ASM programs.
Of course, they can be simulated using the other ASM instructions
(so, they are syntactic sugar)
but it may be convenient to be able to explicitly tell
there is a failure
when something like a division by zero is to be done.
This is what is done in many programming languages with the
so-called exceptions.
Observe that $\fail$ has some common flavor with $\undef$.
However, $\fail$ is relative to executions of programs whereas
$\undef$ is relative to the universe on which the program is
executed.
\\
4. As mentioned in \S\ref{ss:euclid},
considering several outputs goes along with the idea of parallel
updates.
\end{remark}
%
\subsection{Vocabulary and States of an ASM}
%
ASM vocabularies and ASM states correspond to algebraic signatures
and algebras.
The sole difference is that an ASM vocabulary comes with an extra
classification of its symbols as static, dynamic, input and output
carrying the intuitions described in points 2, 3, 6, 7 of
\S\ref{ss:approach}.
\begin{definition} \label{def:vocabularystates}
1. An ASM vocabulary is a finite family of sorts $s_1,\ldots,s_m$
and a finite family $\+L$ of function symbols with specified types
of the form $s_i$ or $s_{i_1}\times\cdots\times s_{i_k} \to s_i$
(function symbols with type $s_i$ are also called constants of type $s_i$).
Four subfamilies of symbols are distinguished:
$$
\begin{array}{lcl}
\+L^\sta \text{ (static symbols)} &\quad,\quad&\+I \text{ (input symbols)}
\\
\+L^\dyn \text{ (dynamic symbols)}& \quad,\quad&\+O \text{ (output symbols)}
\end{array}
$$
such that $\+L^\sta, \+L^\dyn$ is a partition of $\+L$
and $\+I\subseteq\+L^\sta$ and $\+O \subseteq\+L^\dyn$.
We also require that there is a sort to represent Booleans
and that $\+L^\sta$ contains symbols to represent the Boolean items
(namely symbols $\true$, $\false$, $\neg$, $\wedge$, $\vee$)
and, for each sort $s$, a symbol $=_{s}$ to represent equality on sort $s$.
\medskip\\
2. Let $\+L$ be an ASM vocabulary with $n$ sorts.
An {\em  $\+L$-state} is any $n$-sort multialgebra $\+S$
for the vocabulary $\+L$.
The multi-domain of $\+S$ is denoted by $(\+U_1,\ldots,\+U_m)$.
We require that
\begin{enumerate}
\item[i.]
one of the $\+U_i$'s is $\Bool$ with the expected interpretations of
symbols $\true$, $\false$,s $\neg$, $\wedge$, $\vee$,
\item[ii.]
the interpretation of the symbol $=_i$ is usual equality
in the interpretation $\+U_i$ of sort $s_i$.
\end{enumerate}
\end{definition}
In the usual way, using variables typed by the $n$ sorts of $\+L$,
one constructs typed $\+L$-terms and their types.
The type of a term $t$ is of the form
$s_i$ or $s_{i_1}\times\cdots\times s_{i_k} \to s_i$
where $s_{i_1},\ldots, s_{i_k}$ are the types of the different variables
occurring in $t$.
Ground terms are those which contain no variable.
The semantics of typed terms is the usual one.
\begin{definition} \label{not:interpretation}
Let  $\+L$ be an ASM vocabulary and $\+S$ an ASM $\+L$-state.
Let $t$ be a typed term with type
$s_{i_1} \times\cdots\times s_{i_1}\to s_i$.
We denote by $t_{\+S}$ its interpretation in $\+S$,
which is a function $\+U_{i_1}\times\cdots\times \+U_{i_\ell}\to \+U_i$.
In case $\ell=0$, i.e., no variable occurs, then $t_{\+S}$ is an
element of $\+U_i$.
\end{definition}
It will be convenient to lift the interpretation of a term
with $\ell$ variables to be a function with any arity $k$
greater than $\ell$.
\begin{definition} \label{def:lift}
Let  $\+L$ be an ASM vocabulary and $\+S$ an ASM $\+L$-state
with universe $\+U$.
Suppose $\sigma:\{1,\ldots,\ell\}\to\{1,\ldots,p\}$ is any map and
$\tau:\{1,\ldots,p\}\to\{1,\ldots,m\}$ is a distribution of (indexes of) sorts.
Suppose $t$ is a typed term of type
$s_{\tau(\sigma(1))}\times\cdots\times s_{\tau(\sigma(\ell))} \to s_i$.
We let $t^{\tau,\sigma}_{\+S}$
be the function
$\+U_{s_{\tau(1)}}\times\cdots\times\+U_{s_{\tau(p)}} \to \+U_i $
such that, for all
$(a_1,\cdots, a_p) \in \+U_{s_{\tau(1)}}\times\cdots\times\+U_{s_{\tau(p)}}$,
\begin{eqnarray*}
t^{\tau,\sigma}_\+S(a_1,\cdots,a_k)&=&
t_{\+S}(a_{\sigma(1)},\cdots,a_{\sigma(\ell)})\ .
\end{eqnarray*}
 \end{definition}
%
%
\subsection{Initialization Maps}
%
$\+L$-terms with no variable
are used to name particular elements in the universe $\+U$ of an ASM
whereas $\+L$-terms
with variables are used to name particular functions over $\+U$.

Using the lifting process described in Definition \ref{def:lift},
one can use terms containing less than $k$ variables
to name functions with arity $k$.
\begin{definition}\label{def:ini}
1, Let $\+L$ be an ASM vocabulary.
An {\em $\+L$-initialization map $\xi$} has domain family $\+L^{(\dyn)}$ of dynamic symbols
and satisfies the following condition:
\begin{quote}
if $\alpha$ is a dynamic function symbol with type
$s_{\tau (1)}\times\cdots\times s_{\tau(\ell)} \to s_i$
then $\xi(\alpha)$ is a pair $(\sigma,t)$ such that
$\sigma:\{1,\ldots,\ell\}\to\{1,\ldots,p\}$
and $t$ is a typed $\+L$-term with type
$s_{\tau(\sigma(1))}\times\cdots\times s_{\tau(\sigma(\ell))} \to s_i$
which is built with the sole static symbols
(with $\tau:\{1,\ldots,p\}\to\{1,\ldots,m\}$).
\end{quote}
2. Let  $\xi$ be an $\+L$-initialization map.
An $\+L$-state $\+S$ is $\xi$-initial if,
for any dynamic function symbol $\alpha $, if $\xi(\alpha)=(\sigma,t)$ then
the interpretation of $\alpha $ in $\+S$
is $t^{\tau,\sigma}_{\+S}$.
\medskip\\
3. An $\+L$-state is initial if it is $\xi$-initial for some $\xi$.
\end{definition}
\begin{remark}
Of course, the values of static symbols are basic ones,
they are not to be defined from anything else:
either they are inputs or they are the elementary pieces upon which
the ASM algorithm is built.
\end{remark}
%
%
%
\subsection{ASM Programs}
%
\begin{definition}\label{def:program}
1. The vocabulary of ASM programs is the family of symbols
\\\centerline{$
                      \{\skipa\ ,\ \halt \ ,\ \fail 
\ ,\    :=
\ ,\   \left|\begin{array}{l}
                              \phantom{x}\\
                              \phantom{x}
                              \end{array}\right.
\hspace{-3mm} , \ifa\ldots\then\ldots\elsea\ldots\}$}
2. {\em ($\+L$-updates).}
Given an ASM vocabulary $\+L$,
a sequence of $k+1$ ground typed $\+L$-terms $t_1,\ldots,t_k,u$
(i.e. typed terms with no variable),
a dynamic function symbol $\alpha$,
if $\alpha(t_1,\ldots,t_k)$ is a typed $\+L$-term with the same type as $u$
then the syntactic object \ $\alpha(t_1,\ldots,t_k) := u $\
is called an $\+L$-update.
\medskip\\
3. {\em ($\+L$-programs)}.
Given an ASM vocabulary $\+L$, the $\+L$ programs are obtained
via the following clauses.
\begin{enumerate}
\item[i.]{\em (Atoms)}.
$\skipa, \halt, \fail$ and  all $\+L$-updates are $\+L$-programs.
\item[ii.]{\em (Conditional constructor)}.
Given a ground typed term $C$ with Boolean type
and two $\+L$-programs $P,Q$, the syntactic object
$$
\ifa  C  \then  P  \elsea Q
$$
is an $\+L$-program.
\item[iii.]{\em (Parallel block constructor)}.
Given $n\geq1$ and $\+L$-programs $P_1,\ldots,P_n$,
the syntactic object (with a vertical bar on the left)
$$
\left | \begin{array}{l}
P_1\\
\vdots\\
P_n
\end{array} \right.
$$
is an $\+L$-program.
\end{enumerate}
\end{definition}
The intuition of programs is as follows.
\begin{itemize}
\item
$\skipa$ is the program which does nothing.
$\halt$ halts the execution in a successful mode
and the outputs are the current values of the output symbols.
$\fail$ also halts the execution but tells that there is a failure,
so that there is no meaningful output.
\item
Updates modify the interpretations of dynamic symbols,
they are the basic instructions. The left member has to
be of the form $\alpha(\cdots)$ with $\alpha$ a dynamic symbol
because the interpretations of static symbols do not vary.
\item
The conditional constructor has the usual meaning 
whereas the parallel constructor is a new control structure
to get {\em simultaneous and independent executions} of programs $P_1,\ldots,P_n$.
\end{itemize}
%
%
\subsection{Action of an $\+L$-Program on an $\+L$-State}
%
\subsubsection{Active Updates and Clashes}
In a program the sole instructions which have some impact are updates.
They are able to modify the interpretations of dynamic symbols
on the sole tuples of values which can be named by tuples
of ground terms.
Due to conditionals, not every update occurring in a program
will really be active. it does depend on the state to which the program
is applied.
Which symbols on which tuples are really active and what is their action?
This is the object of the next definition.

\begin{definition}[Active updates]\label{def:active}
Let $\+L$ be an ASM vocabulary,  $P$ an $\+L$-program
and $\+S$ an $\+L$-state.
Let $\upd(P)$ be the family of all updates occurring in $P$.
The subfamily $\act(\+S,P)\subseteq \upd(P)$ of so-called
{\em$(\+S,P)$-active updates} is defined via the following induction
on $P$ :
$$
\begin{array}{l}
\act(\+S,\skipa)=\emptyset
\\
\act(\+S,\alpha(t_1,\ldots,t_k) := u)
                =\{\alpha(t_1,\ldots,t_k) := u\}
\\
\act(\+S, \ifa C \then Q \elsea R) = 
	\left\{\begin{array}{ll}
	\act(\+S, Q) & \text{if } C_{\+S}=\true\\
	\act(\+S, R) & \text{if } C_{\+S}=\false\\
	\emptyset & \text{if } C_{\+S}\notin\Bool\\
	\end{array}\right.
\\
\act(\+S, \left | \begin{array}{l} P_1\\ \vdots\\ P_n \end{array} \right.)
	=  \act(\+S, P_1)\cup\ldots\cup \act(\+S, P_n)
\end{array}
$$
\end{definition}
The action of a program $P$ on a state $\+S$ is to be seen as
the conjunction of updates in $\act(\+S,P)$
provided these updates are compatible.
Else, $P$ clashes on $\+S$.
\begin{definition}\label{def:clash}
An $\+L$-program $P$ clashes on an $\+L$-state $\+S$ if
there exists two active updates
$\alpha(s_1,\ldots,s_k) := u$
and $\alpha(t_1,\ldots,t_k) := v$ in $\act(\+S,P)$ 
relative to the same dynamic symbol $\alpha$ such that
${s_1}_{\+S} = {t_1}_{\+S}$, \ldots, ${s_k}_{\+S} = {t_k}_{\+S}$
but $u_{\+S}$ and $v_{\+S}$ are not equal
(as elements of the universe).
\end{definition}
\begin{remark}\label{rk:active}
A priori, another case could also be considered as a clash.
We illustrate it for a parallel block of two programs $P,Q$
and the update of a dynamic constant symbol $c$.
Suppose $c_{\+S} \neq u_{\+S}$
and $c:=u$ is an active update in $\act(\+S,P)$.
Then $P$  wants to modify the value of $c_{\+S}$.
Suppose also that there is no active update with left member $c$
in $\act(\+S,Q)$.
Then $Q$ does not want to touch the value of
$c_{\+S}$.
Thus, $P$ and $Q$ have incompatible actions:
$P$ modifies the interpretation of $c$
whereas $Q$ does nothing about $c$.
One could consider this as a clash for the parallel program
$\left | \begin{array}{l} P\\ Q \end{array} \right.$.
Nevertheless, {\em this case is not considered to be a clash}.
A moment reflection shows that this is a reasonable choice.
Otherwise, a parallel block would always clash
except in case all programs $P_1,\ldots,P_n$ do exactly the same
actions...
Which would make parallel blocks useless.
\end{remark}
%
%
\subsubsection{Halt and Fail}
\begin{definition}\label{def:halt}
Let $\+L$ be an ASM vocabulary, $\+S$ be an $\+L$-state
and $P$ an $\+L$-program. 
By induction, we define the two notions:
$P$ halts (resp. fails) on $\+S$.
\begin{itemize}
\item
If $P$ is $\skipa$ or an update then $P$ neither halts nor fails
on $\+S$.
\item
If $P$ is $\halt$ (resp. $\fail$) then $P$ halts and does not fail
(resp. fails and does not halt) on $\+S$.
\item
$\ifa C \then Q \elsea R$\ halts on $\+S$ if and only if
$$
\left\{\begin{array}{l}
\text{either $C_{\+S}=\true$ and $Q$ halts on $\+S$}\\
\text{or $C_{\+S}=\false$ and $R$ halts on $\+S$}
\end{array}\right.
$$
\item
$\ifa C \then Q \elsea R$ fails on $\+S$ if and only if
$$
\left\{\begin{array}{l}
\text{either $C_{\+S}=\true$ and $Q$ fails on $\+S$}\\
\text{or $C_{\+S}=\false$ and $R$ fails on $\+S$\ .}
\end{array}\right.
$$
\item
The parallel block of programs $P_1,\ldots,P_n$ halts
on $\+S$ if and only if some $P_i$ halts on $\+S$ and no $P_j$
fails on $\+U$.
\item
The parallel block of programs $P_1,\ldots,P_n$. fails
on $\+S$ if and only if some $P_i$ fails on $\+S$.
\end{itemize}
\end{definition}
%
%
%
\subsubsection{Successor State}
\begin{definition}\label{def:suc}\em
Let $\+L$ be an ASM vocabulary and $\+S$ be an $\+L$-state. 
\\
The {\em successor state} $\+T = \suc(\+S ,P)$ of state $\+S$
relative to an $\+L$-program $P$ is defined if only if
$P$ does not clash nor fail nor halt on $\+S$.
\\
In that case, the successor is inductively defined via the following clauses. 
\begin{enumerate}
\item
$\+T=\suc(\+S ,P)$ and $\+S$ have the same base sets $\+U_1,\ldots, \+U_n$.
\item
$\alpha_{\+T} = \alpha_{\+S}$ for any static symbol $\alpha$.
\item[3a.]
$\suc(\+S,\skipa)= \+S$
(recall that $\skipa$ does nothing\ldots.)
\item[3b.]
Suppose $P$ is an update program $\alpha(t_1,\ldots,t_k) := u$
where $\alpha$ is a dynamic symbol with type
$s_{i_1}\times\cdots\times s_{i_k} \to s_i$
and $\vec{a} = ({t_1}_{\+S},\ldots, {t_k}_{\+S})$.
Then all dynamic symbols different from $\alpha$ have the same
interpretation in $\+S$ and $\+T$
and, for every $\vec{b}\in\+U_{i_1}\times\cdots\times\+U_{i_k}$, we have
$\alpha_{\+T}(\vec{b}) =
	\left\{\begin{array}{ll}
	 \alpha_{\+S}(\vec{b}) & \text{if } \vec{b} \neq \vec{a}\\
	 u_{\+S} & \text{if } \vec{b} = \vec{a}
	 \end{array}\right.$.
\item[3c.]
Suppose $P$ is the conditional program\
$\ifa  C  \then  Q  \elsea R$.
Then
$$\left\{\begin{array}{ll}
	  \suc(\+S,P) = \suc(\+S,Q) & \text{if } C_{\+S}=\true\\
	  \suc(\+S,P) = \suc(\+S,R) & \text{if }C_{\+S}=\false
	  \end{array}\right.
$$
(since $P$ does not fail on $\+S$,
we know that $C_{\+S}$ is a Boolean).%
\item[3d]
Suppose $P$ is the parallel block program\
$\left | \begin{array}{l}
P_1\\
\vdots\\
P_n
\end{array} \right.$
and $P$ does not clash on $\+S$.
Then $\+T = \suc(\+S,P)$ is such that,
for every dynamic symbol $\alpha$ with type
$s_{i_1}\times\cdots\times s_{i_k} \to s_i$
and every tuple $\vec{a}=(a_1,\ldots,a_k)$ in
$\+U_{i_1}\times\cdots\times\+U_{i_k}$,
\begin{itemize}
\item
if there exists an update $ \alpha(t_1,\ldots,t_k):=u$
in $\act(\+S,P)$
such that $\vec{a}=({t_1}_{\+S},\ldots, {t_k}_{\+S})$ then
$\alpha(\vec{a})_{\+T}$ is the common value of all $v_{\+S}$
for which there exists some update $ \alpha(s_1,\ldots,s_k):= v$
in $\act(\+S,P)$ such that $\vec{a}=({s_1}_{\+S},\ldots, {s_k}_{\+S})$.
\item
Else $\alpha(\vec{a})_{\+T} = \alpha(\vec{a})_{\+S}$.
\end{itemize}
\end{enumerate}
\end{definition}
\begin{remark}
In particular, $\alpha_{\+T}(\vec{a})$ and $\alpha_{\+S}(\vec{a})$
have the same value in case
$\vec{a}=(a_1,\ldots,a_k)$ is not the value in $\+S$
of any $k$-tuple of ground terms $(t_1,\ldots,t_k)$ such that
$\act(\+S,P)$ contains an update of the form
$\alpha(t_1,\ldots,t_k):=u$ for some ground term $u$.
\end{remark}

\subsection{Definition of ASMs and ASM Runs}
%
At last, we can give the definition of ASMs and ASM runs.
\begin{definition}\label{def:ASM}
1. An ASM is a triple\ $(\+L , P , (\xi, \+J))$
(with two morphological components and one
semantico-morphological component)
such that:
\begin{itemize}
\item
$\+L$ is an ASM vocabulary as in Definition \ref{def:vocabularystates},
\item
$P$ is an $\+L$-program as in Definition \ref{def:program},
\item
$\xi$ is an $\+L$-initialization map and $\+J$ is a $\xi$-initial $\+L$-state
as in Definition \ref{def:ini}.
\end{itemize}
An ASM has type $0$ if all its dynamic symbols have arity $0$
(i.e., they are constants).
\medskip\\
2.  The {\em run of an ASM} $(\+L , P , (\xi, \+J))$ is the
sequence of states $(\+S_i)_{i\in I}$
indexed by a finite or infinite initial segment $I$ of $\N$
which is uniquely defined by the following conditions:
\begin{itemize}
\item
$\+S_0$ is $\+J$.
\item
$i+1\in I$ if and only if $P$ does not clash nor fail
nor halt on $\+S_i$
 and $\act(\+S_i,P) \neq \emptyset$
(i.e. there is an active update\footnote{Nevertheless,
it is possible that $\+S_i$ and $\suc(\+S_i, P)$ coincide,
cf. Remark \ref{rk:infiniterun}.}).
\item
If $i+1\in I$ then $\+S_{i+1} = \suc(\+S_i, P)$.
\end{itemize}
3. Suppose $I$ is finite and $i$ is the maximum element of $I$.
\\
The run is successful if $\act(\+S_i,P)$ is empty or $P$ halts on $\+S_i$.
In that case the outputs are the interpretations on $\+S_i$
of the output symbols.
\\
The run fails if $P$ clashes or fails on $\+S_i$.
In that case the run has no output.
\end{definition}
\begin{remark}\label{rk:infiniterun}
In case $\act(\+S_i,P) \neq \emptyset$
and $P$ does not clash nor fail nor halt on $\+S_i$
and $\+S_i=\+S_{i+1}$
(i.e., if the active updates do not modify $\+S_i$)
then the run is infinite: $\+S_j=\+S_i$ for every $j>i$.
\end{remark}
%
%
\subsection{Operational Completeness: the ASM Theorem}\label{ss:sequentialthm}
%
Let us now state the fundamental theorem of ASMs.
\begin{theorem}
[ASM Theorem, 1999 \cite{gurevichSeqThesis99,gurevichSeqThesis00},
cf. \cite{dershowitzgurevich08}]
Every process satisfying the Sequential Postulates (cf. \S\ref{ss:Thesis})
can be emulated by an ASM with the same vocabulary, sets of states
and initial states.
\end{theorem}
In other words, using Gurevich Sequential Thesis \ref{Thesis}, 
every sequential algorithm can be step by step emulated by an ASM
with the same values of all environment parameters.
I.e., ASMs are operationally complete as concerns sequential
algorithms.
\medskip

The proof of the ASM Theorem also shows that ASM programs of a
remarkably simple form are sufficient.
\begin{definition}\label{def:equivalent}
Let $\+L$ be an ASM vocabulary.
Two ASM $\+L$-programs $P,Q$ are equivalent if,
for every $\+L$-initialization map $\xi$
and every $\xi$-initial state $\+J$,
the two ASMs $(\+L , P , (\xi, \+J))$ and $(\+L , Q , (\xi, \+J))$
have exactly the same runs.
\end{definition}
\begin{theorem}[Gurevich, 1999 \cite{gurevichSeqThesis99}]
\label{thm:normalprogram}
Every ASM program is equivalent to a program which is a parallel block
of conditional blocks of updates, halt or fail instructions,
namely a program of the form:
$$
\left|\begin{array}{l}
\ifa C_1 \then
	\left|\begin{array}{l}
	I_{1,1}\\
	\vdots\\
	I_{1,p_1}
	\end{array}\right.
\\
\vdots\\
	\ifa C_n \then
	\left|\begin{array}{l}
	I_{n,1}\\
	\vdots\\
	I_{n,p_n}
	\end{array}\right.
\end{array}\right.
$$
where the $I_{i,j}$'s are updates or $\halt$ or $\fail$
and the interpretations of $C_1$,\ldots, $C_n$ in any state
are Booleans such that at most one of them is $\true$.
\end{theorem}

\begin{proof}
For $\skipa,\halt,\fail$ consider an empty parallel block.
For an update or $\halt$ or $\fail$ consider a block of
one conditional with a tautological condition.
Simple Boolean conjunctions allow to transform a conditional of
two programs of the wanted form into the wanted form.
The same for parallel blocks of such programs.
\end{proof}
%
%
%
%
%
\section{Lambda Calculus}
\label{s:lambda}
%
As much as possible, our notations are taken from Barendregt's book
\cite{barendregt}
(which is a standard reference on $\Lambda$-calculus).
%
\subsection{Lambda Terms}\label{ss:lambdaterms}
%
Recall that the family $\Lambda$  of $\lambda$-terms of the
$\Lambda$-calculus is constructed from an infinite family of variables
via the following rules:
\begin{enumerate}
\item
Any variable is a $\lambda$-term.
\item
{\em (Abstraction)}
If $x$ is a variable and $M$ is a $\lambda$-term then
$\lambda x \dotsp M$ is a
$\lambda$-term.
\item
{\em (Application)}
If $M,N$ are $\lambda$-terms then $(M\ N)$ is a $\lambda$-term.
\end{enumerate}

Free and bound occurrences of a variable in a $\lambda$-term are defined
as in logical formulas, considering that abstraction $\lambda x \dotsp M$
bounds $x$ in $M$.

One considers $\lambda$-terms  up to a renaming
(called $\alpha$-conversion) of their bound variables.
In particular, one can always suppose that, within a $\lambda$-term,
no variable has both free occurrences and bound occurrences 
and that any two abstractions  involve distinct variables.

\medskip
To simplify notations, it is usual to remove parentheses in terms,
according to the following conventions:
\begin{itemize}
\item
applications associate leftwards: in place of
$(\cdots( (N_1\ N_2)\ N_3)\cdots\ N_k)$ we write
$N_1 N_2 N_3 \cdots N_k$,
\item
abstractions associate rightwards:
$\lambda x_1 \dotsp (\lambda x_2 \dotsp 
                            (\cdots \dotsp (\lambda x_k . M)\cdots))$
is written $\lambda x_1\cdots x_k \dotsp M$.
\end{itemize}

%
\subsection{$\beta$-Reduction} \label{ss:beta}
%
\begin{note}
Symbols $:=$ are used for updates in ASMs and are also commonly
used in $\Lambda$-calculus to denote by $M[x:=N]$ the substitution
of all occurrences of a variable $x$ in a term $M$ by a term $N$.
To avoid any confusion, we shall rather denote such a substitution
by $M[N/x]$.
\end{note}

\medskip
\begin{figure}\center
\begin{tabular}{|rccccl|}
\hline
&\multicolumn{4}{c}
{\bf Decorated rules of reduction in $\mathbf{\Lambda}$-calculus}&
\\ \hline
\multicolumn{2}{|c|}{} & \multicolumn{4}{c|}{}\\
(Id) & \multicolumn{1}{c|}{$M \redu_0 M$}
& \multicolumn{3}{c}{$(\lambda x.M)\ N \redu_1 M[N/x]$} & $(\beta)$
\\&& \multicolumn{4}{|c|}{}
\\\hline \multicolumn{3}{|c|}{} & \multicolumn{3}{c|}{}\\
(App) & \multicolumn{2}{c|}
{\AxiomC{$M \redu_i M'$}
 \UnaryInfC{$MN \redu_i M'N$}
\DisplayProof
\   \AxiomC{$N \redu_i N'$}
    \UnaryInfC{$MN \redu_i MN'$}
    \DisplayProof}
&\multicolumn{2}{c}
{\AxiomC{$M\redu_i M'$}
\UnaryInfC{$(\lambda x.M)\redu_i \lambda x.M'$}
\DisplayProof}
&\!\!\!\!(Abs)
\\ \multicolumn{3}{|c|}{} & \multicolumn{3}{c|}{}
\\ \hline

\end{tabular}
\caption{Reductions with decorations}
\label{tab:red}
\end{figure}

The family of $\lambda$-terms is endowed with a
{\em reducibility relation},  called $\beta$-reduction and
denoted by $\redu$.
\begin{definition}
1. Let $P$ be a $\lambda$-term.
A subterm of $P$ the form $(\lambda x. M) N$
is called a $\beta$-redex (or simply redex) of $P$.
Going from $P$ to the $\lambda$-term $Q$ obtained by
substituting in $P$ this redex by $M[N/x]$
(i.e., substituting $N$ to every free occurrence
of $x$ in $M$) is called a $\beta$-reduction and we write
\ $  P\ \redu\ Q\ .$
\medskip\\
2. The iterations $\redu_i$ of $\redu$
and the reflexive and transitive closure $\redustar$
are defined as follows:
$$
\begin{array}{lcl}
\redu_0 &=& \{(M,M) \mid M\}
\\
\redu_{i+1} &=& \redu_i \circ \redu
\text{\qquad (so that $\redu\ =\ \redu_1$)}
\\
&=&\{(M_0,M_i) \mid \exists M_1,\ldots,M_i \mid
                         M_0 \redu M_1 \redu \cdots \redu M_i \redu M_{i+1}\}
\\
\redustar &=& \bigcup_{i\in\N} \redu_i
\end{array}
$$
These reduction relations are conveniently expressed
via axioms and rules (cf. Figure~1):
the schema of axioms $(\beta)$ gives the core transformation
whereas rules (App) and (Abs) insure that this can be done
for subterms.
\end{definition}
Relations $\redu_i$ are of particular interest to analyse the
complexity of the simulation of one ASM step in $\Lambda$-calculus.
Observe that axioms and rules for $\redu$ extend to $\redustar$.
%
%
\subsection{Normal Forms}\label{ss:normalforms}
%
\begin{definition}\label{def:normal}
1. A $\lambda$-term $M$ is in normal form if it contains
no redex.
\medskip\\
2. A $\lambda$-term $M$ has a normal form if there exists some
term $N$  in normal form such that $M\redustar N$.
\end{definition}
\begin{remark}\label{rk:confluent}
There are terms with no normal form.
The classical example is $\Omega=\Delta\Delta$ where
$\Delta=\lambda x\dotsp xx$. Indeed, $\Omega$ is a redex and
reduces to itself.
\end{remark}
In a $\lambda$-term, there can be several subterms
which are redexes,
so that iterating $\redu$ reductions is a highly
non deterministic process.
Nevertheless, going to normal form is a functional process.
\begin{theorem}[Church-Rosser \cite{churchrosser}, 1936]
\label{thm:churchrosser}
The relation $\redustar$ is confluent:
if $M\redustar N'$ and $M\redustar N''$ then there exists $P$
such that $N'\redustar P$ and $N''\redustar P$.
In particular, there exists at most one term $N$ in normal form
such that $M\redustar N$.
\end{theorem}
\begin{remark}\label{rk:redinotconfluent}
Theorem \ref{thm:churchrosser} deals with $\redustar$ exclusively:
relation $\redu_i$ is {\em not} confluent for any $i\geq1$.
\end{remark}
A second fundamental property is that going to normal form
can be made a deterministic process.
\begin{definition}\label{def:leftmost}
Let $R',R''$ be two occurrences of redexes in a term $P$.
We say that $R'$ is left to $R''$ if the first lambda in $R'$
is left to the first lambda in $R''$ (all this viewed in $P$).
If terms are seen as labelled ordered trees, this means that
the top lambda in $R'$ is smaller than that in $R''$
relative to the prefix ordering on nodes of the tree $P$.
\end{definition}

\begin{theorem}[Curry \& Feys \cite{curry1958}, 1958]
\label{thm:leftmost}
Reducing the leftmost redex of terms not in normal form
is a deterministic strategy which leads to the normal form
if there is some.
\\
In other words, if $M$ has a normal form $N$
then the sequence $M=M_0\redu M_1\redu M_2\redu\cdots$
where each reduction $M_i\redu M_{i+1}$ reduces the leftmost redex
in $M_i$ (if $M_i$ is not in normal form)
is necessarily finite and ends with $N$.
\end{theorem}
%
\subsection{Lists in $\Lambda$-Calculus}\label{ss:lists}
%
We recall the usual representation of lists in $\Lambda$-calculus
with special attention to decoration
(i.e., the number of $\beta$-reductions in sequences of reductions).
\begin{proposition}
Let $\langle u_1,\ldots,u_k\rangle = \lambda z \dotsp z u_1\ldots u_k$ and,
for $i=1,\ldots,k$, let
$\pi^k_i=\lambda x_1\ldots x_k \dotsp x_i$.
Then
$\langle u_1,\ldots,u_k\rangle\ \pi^k_i \redu_{1+k}\ u_i$.
\\
Moreover, if all $u_i$'s are in normal form then so is
$\langle u_1,\ldots,u_k\rangle$ and these reductions
are deterministic: there exists a unique sequence of reductions
from $\langle u_1,\ldots,u_k\rangle$ to $u_i$.
\end{proposition}
%
%
\subsection{Booleans in $\Lambda$-Calculus}\label{ss:bool}
%
We recall the usual representation of Booleans in $\Lambda$-calculus.
\begin{proposition}
Boolean elements $\true, \false$ and usual Boolean functions
can be represented by the following $\lambda$-terms, all in normal form:
$$
\begin{array}{|c|}
\hline
\begin{array}{rcl}
\corner{\true} &=& \lambda xy.x\\
\corner{\false} &=& \lambda xy.y
\end{array}
\qquad
\begin{array}{rcl}
\text{\tt neg} &=& \lambda x\dotsp x \corner{\false}\ \corner{\true}\\
\text{\tt and} &=& \lambda xy \dotsp xy \corner{\false}\\
\text{\tt or} &=& \lambda xy \dotsp x\ \corner{\true}\ y\\
\text{\tt implies} &=& \lambda xy \dotsp xy\corner{\true}\\
\text{\tt iff} &=& \lambda xy \dotsp xy(\corner{\neg}y)
\end{array}
\\\hline
\end{array}
$$
For $a,b\in\{\true,\false\}$, we have
$\text{\tt neg}\ \corner{a} \redu \corner{\neg a}$,
$\text{\tt and}\ \corner{a}\ \corner{b} \redustar \corner{a f b}$,\ldots.
\end{proposition}
\begin{proposition}[If Then Else]\label{p:ifthenelse}
For all terms $M,N$,
$$
(\lambda z \dotsp zMN)\ \corner{\true}   \redu_2  M
\quad,\quad
(\lambda z \dotsp zMN)\ \corner{\false}  \redu_2  N\ .
$$
\end{proposition}
We shall use the following version of iterated conditional.
\begin{proposition}\label{p:case}
For every $n\geq 1$ there exists a term $\Case_n$
such that, for all normal terms $M_1,\ldots,M_n$ and all
$t_1,\ldots,t_n\in\{\corner{\true}, \corner{\false}\}$,
$$
\Case_n\ M_1\ldots M_n\  t_1\ldots t_n   \redu_{3n}  M_i
$$
relative to leftmost reduction in case
$t_i= \corner{\true}$ and $\forall j<i\ t_j= \corner{\false}$.
\end{proposition}
\begin{proof}
Let $u_i = y_i(\lambda x_{i+1}\dotsp I)\ldots(\lambda x_n\dotsp I)$,
set
\begin{eqnarray*}
\Case_n&=& \lambda y_1\ldots y_n z_1\ldots z_n \dotsp
z_1 u_1(z_2 u_2(\ldots (z_{n-1}u_{n-1}(z_n u_n I))\ldots))
\end{eqnarray*}
and observe that, for leftmost reduction, letting $M'_i = u_i[M_i/y_i]$,
\begin{eqnarray*}
\Case_n\ M_1\ldots M_n\ t_1\ldots t_n& \redu_{2n}&
t_1 M'_1(t_2 M'_2(\ldots (t_{n-1}M'_{n-1}(t_n M'_n I))\ldots))
\\
& \redu_i& M'_i
\\
& \redu_{n-i}& M_i\ .
\end{eqnarray*}
\end{proof}
%
%
\subsection{Integers in $\Lambda$-Calculus}\label{ss:integers} 
%
There are several common representations of integers in
$\Lambda$-calculus. We shall consider a slight variant of the standard
one (we choose another term for $\corner{0}$),
again with special attention to decoration.
\begin{proposition}
Let
$$
\begin{array}{|rclcl|}
\hline
\corner{0} &=&
\lambda z\dotsp z \corner{\true} \corner{\false}&&
\\
\corner{n+1} &=&
\langle \corner{\false}, \corner{n} \rangle
&=& \lambda z \dotsp z \corner{\false} \corner{n}
\\
\slzero &=& \lambda x\dotsp x \corner{\true}&&\\
\slsuc &=& \lambda z\dotsp \langle \corner{\false}, z\rangle&&\\
\slpred &=& \lambda z\dotsp x \corner{\false}&&
\\ \hline
\end{array}
$$
The above terms are all in normal form and
$$
\begin{array}{rcl}
\slzero \corner{0} &\redu_3& \corner{\true}\\
\slzero \corner{n+1} &\redu_3& \corner{\false}
\end{array}
\quad , \quad
\begin{array}{rcl}
\slsuc \corner{n} &\redu_3& \corner{n+1}\\
\slpred \corner{n+1} &\redu_3& \corner{n}\\
\slpred \corner{0} &\redu_3& \corner{\false}
\end{array}\ .
$$
Moreover, all these reductions are deterministic.
\end{proposition}
\begin{remark}
The standard definition sets
$\ulcorner 0 \urcorner = \lambda x\dotsp x$.
Observe that
$\slzero (\lambda x\dotsp x) \redu_2 \corner{\true} $.
The chosen variant of $\ulcorner 0 \urcorner$ is to get the same
decoration (namely $3$)
to go from $\slzero \corner{0} $ to $\corner{\true} $
and to go from $\slzero \corner{n+1} $ to $\corner{\false} $.
\end{remark}
Let us recall Kleene's fundamental result.
\begin{theorem}[Kleene, 1936]
For every partial computable function $f:\N^k\to\N$ there exists
a $\lambda$-term $M$ such that, for every tuple $(n_1,\cdots,n_k)$,
\begin{itemize}
\item
$M \corner{n_1}\cdots \corner{n_k}$
admits a normal form (i.e., is $\redustar$ reducible to a term in normal form)
if and only if $(n_1,\cdots,n_k)$ is in the domain of $f$,
\item
in that case,
$M \corner{n_1}\cdots \corner{n_k}
\redustar \corner{f(n_1,\cdots,n_k)}$
(and, by Theorem \ref{thm:churchrosser}, this normal form is unique).
\end{itemize}
\end{theorem}
%
%
%
\subsection{Datatypes in $\Lambda$-Calculus}
\label{ss:datastructure}
%
We just recalled some representations of Booleans and integers
in $\Lambda$-calculus. In fact, any inductive datatype
can also be represented.
Using computable quotienting, this allows to also represent
any  datatype used in algorithms.
\\
Though we will not extend on this topic, let us recall Scott encoding 
of inductive datatypes in the $\Lambda$-calculus
(cf. Mogensen \cite{mogensen}).
\begin{quote}\em
1. If the inductive datatype has constructors $\psi_1,\ldots, \psi_p$
having arities $k_1,\ldots,k_p$,
constructor $\psi_i$ is represented by the term
$$
\lambda x_1\ldots x_{k_i} \alpha_1 \ldots \alpha_p \dotsp
\alpha_i x_1\ldots x_{k_i}\ .
$$
In particular, if $\psi_i$ is a generator (i.e., an arity $0$ constructor)
then it is represented by the projection term
$\lambda \alpha_1 \ldots \alpha_p \dotsp \alpha_i$.
\\
2. An element of the inductive datatype is a composition of
the constructors and is represented by the similar composition
of the associated $\lambda$-terms.
\end{quote}
Extending the notations used for Booleans and integers,
we shall also denote by $\corner{a} $ the $\lambda$-term
representing an element $a$ of a datatype.
\medskip\\
Scott's representation of inductive datatypes extends to finite families
of datatypes defined via mutual inductive definitions.
It suffices to endow constructors with types and to restrict compositions
in point 2 above to those respecting constructor types.
%
%
%
\subsection{Lambda Calculus with Benign Constants}\label{ss:constants}
%
We consider an extension of the lambda calculus with constants
to represent particular computable functions and predicates.
Contrary to many $\lambda\delta$-calculi
(Church $\lambda\delta$-calculus, 1941 \cite{church41},
Statman, 2000 \cite{statman2000},
Ronchi Della Rocca, 2004 \cite{ronchi2004},
Barendregt \& Statman, 2005 \cite{barendregtstatman}),
this adds no real additional power:
it essentially allows for shortcuts in sequences of reductions.
The reason is that axioms in Definition \ref{def:lambdaconstants}
do not apply to all terms but only to codes of elements
in datatypes. 

\begin{definition}\label{def:lambdaconstants}
Let $\F $ be a family of functions with any arities
over some datatypes $A_1,\ldots,A_n$.
The  $\Lambda_{\F}$-calculus is defined as follows:
\begin{itemize}
\item
The family of $\lambda_{\F}$-terms
is constructed as in \S\ref{ss:lambdaterms} from the family
of variables augmented with constant symbols:
one constant $c_f$ for each $f\in\F$.
\item
The axioms and rules of the top table of Figure \ref{tab:red} are
augmented with the following axioms:
if $f:A_{i_1}\times\cdots\times A_{i_k} \to A_i$
is in $\F$ then, for all $(a_1,\cdots, a_k)\in A_{i_1}\times\cdots\times A_{i_k}$,
$$
(Ax_f)\qquad\qquad c_f\ \corner{a_1} \cdots \corner{a_k}
\ \redu\ \corner{f(a_1,\cdots, a_k)}\ .
$$
\end{itemize}
\end{definition}
\begin{definition}\label{not:delta}
1. We denote by $\betared$ the classical $\beta$-reduction
(with the contextual rules (Abs), (App))
extended to terms of $\Lambda_{\F}$.
\\
2. We denote by $\redu_{\F}$ the reduction given by the
sole $(Ax_f)$\text{-axioms} and the contextual rules (Abs), (App).
\\
3. We use double decorations:
$M\ \redu_{i,j} N$ means that there is a sequence consisting of
$i$ $\beta$-reductions and $j$ ${\F}$-reductions which goes
from $t$ to $u$.
\end{definition}
The Church-Rosser property still holds.
\begin{proposition}\label{p:churchrosser}
The  $\Lambda_{\F}$-calculus is confluent
(cf. Theorem \ref{thm:churchrosser}).
\end{proposition}
\begin{proof}
Theorem \ref{thm:churchrosser} insures that $\betaredstar$
is confluent.
It is immediate to see that any two applications of the $\F$
axioms can be permuted: this is because two distinct $\F$-redexes
in a term are always disjoint subterms.
Hence $\redu_\F$ is confluent.
Observe that $\redustar$ is obtained by iterating finitely many times
the relation $\betaredstar \cup \redu_{\F}$.
Using Hindley-Rosen Lemma
(cf. Barendregt's book \cite{barendregt}, Proposition 3.3.5,
or Hankin's book \cite{hankin}, Lemma 3.27), to prove that $\redustar$ is confluent, it suffices to
prove that $\betaredstar$ and $\redu_{\F}$ commute.
One easily reduces to prove that $\betared$ and $\redu_{\F}$
commute, i.e.,
$$
\exists P\ (M\ \betared P\ \redu_{\F} N)
\quad\iff \quad
\exists Q\ (M\ \redu_{\F} Q\ \betared N)\ .
$$
Any length two such sequence of reductions involves
two redexes in the term $M$:
a $\beta$-redex $R = (\lambda x\dotsp A)B$
and a ${\F}$-redex
$C = c\ \corner{a_1} \cdots \corner{a_k} $.
There are three cases: either $R$ and $C$ are disjoint subterms of $M$
or $C$ is a subterm of $A$ or $C$ is a subterm of $B$.
Each of these cases is straightforward.
\end{proof}
We adapt the notion of leftmost reduction in the
$\Lambda_{\F}$-calculus as follows.
\begin{definition}\label{def:leftmostdelta}
The leftmost reduction in $\Lambda_{\F}$
reduces the leftmost ${\F}$-redex if there is some
else it reduces the leftmost $\beta$-redex.
\end{definition}
%
%
%
\subsection{Good $\F$-Terms}\label{ss:Fterms}
%
To functions which can be obtained by composition from functions in $\F$
we associate canonical terms in $\Lambda_{\F}$ and datatypes.
These canonical terms are called good $\F$-terms,
they contain no abstraction,
only constant symbols $c_f$, with $f\in\F$, and variables.
\begin{problem}\label{rk:typepb}
We face a small problem. Functions in $\F$ are to represent static functions
of an ASM.
Such functions are typed whereas $\Lambda_{\F}$ is an untyped lambda calculus.
In order to respect types when dealing with composition of functions in $\F$,
the definition of good $\F$-terms is done in two steps: the first step involves
typed variables and the second one replaces them by untyped variables.
\end{problem}
\begin{definition}\label{def:Fterms}
1. Let  $A_1,\ldots,A_n$ be the datatypes involved in functions 
of the family $\F $.
Consider typed variables $x^{A_i}_j$ where $j\in\N$ and
$i=1,\ldots,n$.
The family of pattern $\F$-terms, their types and semantics
are defined as follows:
Let $f\in\F$ be such that $f:A_{i_1}\times\cdots\times A_{i_k} \to A_q$.
\begin{itemize}
\item
If $x^{A_{i_1}}_{j_1},\ldots,x^{A_{i_k}}_{j_k}$ are typed variables
then the term $c_f\ x^{A_{i_1}}_{j_1}\ldots x^{A_{i_k}}_{j_k}$
is a pattern $\F$-term
with type $A_{i_1}\times\cdots\times A_{i_k} \to A_q$ and semantics
$\means{c_f\ x^{A_{i_1}}_{j_1}\ldots x^{A_{i_k}}_{j_k}}=f$.
\item
For $j=1,\ldots,k$, let $t_j$ be a pattern $\F$-term with datatype $A_j$
or a typed variable $x^{A_j}_i$.
Suppose the term $t = c_f\ t_1\cdots t_k$ contains exactly
the typed variables $x^{A_i}_j$ for $(i,j)\in I$ and,
for $\ell=1,\ldots,k$, the term $t_\ell$ contains exactly
the typed variables $x^{A_i}_j$ for $(i,j)\in I_j\subseteq I$.
\\
Then the term $c_f\ t_1\cdots t_k$ is a pattern $\F$-term with type
$\prod_{i\in I}A_i \to A_q$
and a semantics $\means{c_f\ t_1\cdots t_k}$ such that,
for every tuple $(a_i)_ {i\in I}\in \prod_{i\in I}A_i$,
$$
\means{t}((a_i)_{i\in I})
= f(\means{t_1}((a_i)_{i\in I_1})),
                \ldots,\means{t_k}((a_k)_{i\in I_k})))\ . 
$$
\end{itemize}
2. Good $\F$-terms are obtained by substituting in a pattern $\F$-term
untyped variables to the typed variables
so that two distinct typed variables are substituted
by two distinct untyped variables.
\end{definition}
The semantics of good $\F$-terms is best illustrated by the following
example:
the function $h$ associated to the term
$c_g (c_hy) x (c_gzzx)$ is the one given by
equality  $f(x,y,z) = g(h(y),x, g(z,z,x))$ which corresponds to
Figure \ref{fig:tree}.
\begin{figure}
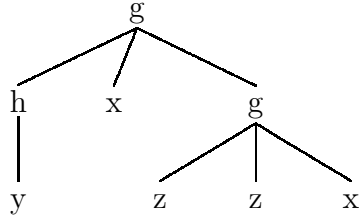
\label{fig:tree}\center
{\synttree{4} 
[g 
    [h [y] ]
     [x]
     [g [z] [z] [x] ]
] }
\caption{Composition tree}
\end{figure}
\\
The reason for the above definition is the following simple result
about reductions of good terms obtained via substitutions.
It is proved via a straightforward induction on good $\F$-terms
and will be used in \S\ref{ss:constantcost},
\ref{ss:constantcostconditional}
\begin{proposition}\label{p:Fterms}
Let $t$ be a good $\F$-term with $k$ variables $y_1,\ldots,y_k$
such that $\means{t}=f:A_{i_1}\times\cdots\times A_{i_k} \to A_q$.
Let $N$ be the number of nodes of the tree
associated to the composition of functions in $\F$ giving $f$
(cf. Figure \ref{fig:tree}). 
\\
There exists $L_t=O(N)$ such that,
for every $(a_1,\ldots,a_k)\in A_{i_1}\times\cdots\times A_{i_k}$,
$$
t[\corner{a_1}/y_1,\ldots,\corner{a_k}/y_k]
\redustar_\F \corner{f(a_1,\ldots,a_k)}
$$
and, using the leftmost reduction strategy,
this sequence of reductions consists of exactly
$L_t$ ${\F}$-reductions.
\end{proposition}
%
%
%
%
%
\section{Variations on Curry's Fixed Point}
\label{s:variations}
%
%
\subsection{Curry's Fixed Point}\label{ss:fixed}
%
Let us recall Curry's fixed point.
\begin{definition}\label{def:curry}
The Curry operator $\varphi \mapsto\theta_\varphi$
on $\lambda$-terms is defined as follows
$$
\theta_F = (\lambda x\dotsp F(xx))(\lambda x\dotsp F(xx))\ .
$$
\end{definition}
\begin{theorem}[Curry's fixed point]\label{thm:curry}
For every $\lambda$-term $F$,
$\theta_F \redu F \theta_F$.
\end{theorem}
\begin{proof}
One $\beta$-reduction suffices: $\theta_F $ is of the form $XX$
and is itself a redex (since $X$ is an abstraction)
which $\beta$-reduces to $F(XX)$, i.e., to $F \theta_F $.
\end{proof}
%
%
%
\subsection{Padding Reductions}\label{ss:padding}
%
We show how to pad leftmost reduction sequences so as to get
prescribed numbers of $\beta$ and ${\F}$-reductions.
\begin{lemma}[Padding lemma]\label{l:padding}
Suppose that $\F$ contains
some function $\omega:B_1\times\cdots\times B_\ell\to B_i$
(with $1\leq i\leq \ell$)
and some constants $\nu_1\in B_1$, \ldots, $\nu_\ell\in B_\ell$.
\\
1. For every $K\geq 2$ and $L\geq0$,
there exists a $\lambda$-term $\pad_{K,L}$ in $\Lambda_{\F}$
with length $O(K+L)$ such that,
for any finite sequence of $\lambda$-terms
$\theta,t_1,\ldots,t_k$ in $\Lambda_{\F}$
which contain no ${\F}$-redex,
\begin{itemize}
\item[i.]
$\pad_{K,L}\ \theta\ t_1 \cdots t_k\ \redustar\ \theta\ t_1 \cdots t_k$.
\item[ii.]
The leftmost derivation consists of exactly
$L$~${\F}$-reductions followed by $K$~$\beta$-reductions.
\end{itemize}
2. Moreover, if $K\geq3$,
one can also suppose that $\pad_{K,L}$ contains
no ${\F}$-redex.
\end{lemma}
\begin{proof}
1. For the sake of simplicity, we suppose that $\omega$
has arity $1$, the general case being a straightforward extension.
Let $I=\lambda x \dotsp x$ and 
$I^\ell = I\cdots I$ ($\ell$ times $I$).
Observe that $I^\ell\ s_0\cdots s_p\ \redustar\ s_0\cdots s_p$
and the leftmost derivation consists of exactly
$\ell $ $\beta$-reductions.
So it suffices to set\ $\pad_{K,0} = I^K$\ and, for $L\geq1$,
$$
\pad_{K,L} = I^{K-2}\ (\lambda xy \dotsp y)\
(\overbrace{\corner{\omega}
                   (\ldots
                   (\corner{\omega}}^{L\text{ times}} \corner{\nu_1})
                                      \ldots))\ .
$$
2. To suppress the ${\F}$-redex $\corner{\omega} \corner{\nu_1}$,
modify $\pad_{K,L}$ as follows:
$$
\pad_{K,L} = I^{K-3}\ (\lambda xy \dotsp xy)\
((\lambda z\dotsp (\overbrace{\corner{\omega}
                   (\ldots
                   (\corner{\omega}}^{L\text{ times}} z)
                                      \ldots)))\ \corner{\nu_1})\ .
$$

\end{proof}
%
%
\subsection{Constant Cost Updates}\label{ss:constantcost}
%
We use Curry's fixed point Theorem and the above padding
technique to insure constant length reductions for any
given update function for tuples.
\begin{lemma}\label{l:curryvaria1}
Let  $A_1,\ldots,A_n$ be the datatypes involved in functions 
of the family $\F $.
Suppose that $\F$ contains
some function $\omega:B_1\times\cdots\times B_\ell\to B_i$
(with $1\leq i\leq \ell$)
and some constants $\nu_1\in B_1$, \ldots, $\nu_\ell\in B_\ell$.
Let $\tau:\{1,\ldots,k\}\to\{1,\ldots,n\}$
be a distribution of indexes of sorts.
For $j=1,\ldots,k$, let $\varphi_j$ be a good $\F$-term
with variables $x_i$ for $i\in I_j\subseteq\{1,\ldots,k\}$
such that~ $\means{\varphi_j}=f_j:\prod_{i\in I_j}A_{\tau(i)} \to A_{\tau(j)}$.
\\
There exists constants $\Kmin$ and $\Lmin$ such that,
for all $K\geq \Kmin $ and $L\geq \Lmin $,
there exists a $\lambda$-term $\theta$ such that,
\begin{enumerate}
\item[1.]                
Using the leftmost reduction strategy,
for all $(a_1,\ldots,a_k)\in A_{\tau(i)}\times\cdots\times A_{\tau(k)}$,
denoting by $\vec{a}_I$ the tuple $(a_j)_{j\in I}$,
\begin{eqnarray}\label{eq:curry1}
\theta\ \corner{a_1}\cdots \corner{a_k} 
&\redustar& \theta\  \corner{f_1(\vec{a}_{I_1})} \cdots
                               \corner{f_k(\vec{a}_{I_k})}\ .
\end{eqnarray}
\item[2.]
This sequence of reductions consists of
$K$ $\beta$-reductions and $L$ ${\F}$-reductions.
\end{enumerate}
\end{lemma}
\begin{proof}
Let $K',L'$ be integers to be fixed later on.
Set
$$
F = \pad_{K',L'}\ \lambda \alpha x_1\ldots x_k
\dotsp  \alpha \varphi_1 \ldots \varphi_k
\qquad
\theta = (\lambda z\dotsp F(zz))\ (\lambda z\dotsp F(zz))\ .
$$
Since $\theta$ and the $\varphi_i$'s have no ${\F}$-redex,
we have the following leftmost reduction:
$$
\begin{array}{rll}
\theta\ \corner{a_1}\cdots \corner{a_k}
& \redu_{1,0}&
F\ \theta\ \corner{a_1}\cdots \corner{a_k}
\text{\qquad (cf. Theorem \ref{thm:curry})}
\\
&=&
\pad_{K',L'}\ (\lambda \alpha x_1\ldots x_k
\dotsp  \alpha \varphi_1 \ldots \varphi_k)\
\theta\ \corner{a_1}\cdots \corner{a_k}
\\
&\redu_{K',L'}&
(\lambda \alpha x_1\ldots x_k
\dotsp  \alpha \varphi_1 \ldots \varphi_k)\
\theta\ \corner{a_1}\cdots \corner{a_k}
\\
&&\text{ \ (apply Lemma \ref{l:padding})}
\\
& \redu_{k+1,0}&
\theta\ \varphi_1[\corner{a_1}/x_1,\ldots, \corner{a_k}/x_k] 
\\&& \qquad \qquad
\cdots \varphi_k[\corner{a_1}/x_1,\ldots, \corner{a_k}/x_k]
\\
& \redu_{0,S}&
\theta\  \corner{f_1(\vec{a}_{I_1})} \cdots
                               \corner{f_k(\vec{a}_{I_k})}
\\
&&\text{ \ (apply Proposition \ref{p:Fterms})}
\end{array}
$$
where $S=\sum_{j=1,\ldots,k}L_{\varphi_j}$.
The total cost is $K'+k+2$ $\beta$-reductions plus
$L'+S$ ${\F}$-reductions.
We conclude by setting $K'=K-(k+2)$
and $L'=L-S$.
\end{proof}
%
%
%
\subsection{Constant Cost Conditional Updates}
\label{ss:constantcostconditional}
%
We refine Lemma \ref{l:curryvaria1} to conditional updates.
\begin{lemma}\label{l:curryvaria2}
Let  $A_1,\ldots,A_n$ be the datatypes involved in functions 
of the family $\F $.
Suppose that $\F$ contains
some function $\omega:B_1\times\cdots\times B_\ell\to B_i$
(with $1\leq i\leq \ell$)
and some constants $\nu_1\in B_1$, \ldots, $\nu_\ell\in B_\ell$.
Let $\tau:\{1,\ldots,k\}\to\{1,\ldots,n\}$, $\iota_1,\ldots, \iota_q\in\{1,\ldots,n\}$
be distributions of indexes of sorts.
Let
$(\rho_s)_{s=1,\ldots,p+q}$,
$(\varphi_{i,j})_{i=1,\ldots,p, j=1,\ldots,k}$,
$(\gamma_\ell)_{i=1,\ldots,q}$
be sequences of good $\F$-terms with variables $x_i$ with $i$ varying
in the respective sets $I_s,I_{i,j},J_\ell\subseteq\{1,\ldots,k\}$.
Suppose that
$$\begin{array}{rcccrcl}
\means{\rho_s} &=& r_s&:& \prod_{i\in I_s}A_{\tau(i)} &\to&\Bool\ ,
\\
\means{\varphi_{i,j}}&=&
f_{i,j}&:&\prod_{i\in I_{i,j}}A_{\tau(i)} &\to& A_{\tau(j)}\ ,
\\
\means{\gamma_\ell} &=&
g_\ell&:&\prod_{i\in J_\ell}A_{\tau(i)} &\to& A_{\iota(\ell)}
\end{array}$$
(in particular, $f_{1,j},\ldots,f_{p,j}$ all take values in $A_{\tau(j)}$).
There exists constants $\Kmin$ and $\Lmin$ such that,
for all $K\geq \Kmin $ and $L\geq \Lmin $,
there exists a $\lambda$-term $\theta$ such that,
\item[1.]
Using the leftmost reduction strategy,
for all $(a_1,\ldots,a_k)\in A_{\tau(1)}\times\cdots\times A_{\tau(k)}$
and $s\in\{1,\ldots,p,p+1,\ldots,p+q\}$,
$$
\begin{array}{ll}
\text{If} & r_s(\vec{a}_{I_s}) = \true\ \wedge\
\forall t<s \ r_t(\vec{a}_{I_t}) = \false
\qquad \qquad (\dagger)_s
\\\\
\text{then}
&  \begin{array}{rcl}
   \theta\ \corner{a_1}\cdots \corner{a_k} 
   &\redustar&
   \left\{\begin{array}{ll}
          \theta\ \corner{f_{s,1}(\vec{a}_{I_{s,1}})}
                  \cdots
                  \corner{f_{s,k}(\vec{a}_{I_{s,k}})}
          &\text{\quad if $s\leq p$}
          \\
           \corner{g_\ell(\vec{a}_{J_\ell})}
            & \text{\quad if $s=p+\ell$}
          \end{array}\right.\ .
    \end{array}
\end{array}
$$
\item[2.]
In all cases, this sequence of reductions consists of exactly
$K$ $\beta$-reductions and $L$ ${\F}$-reductions.
\end{lemma}
\begin{proof}
Let $K',L'$ be integers to be fixed at the end of the proof.
For $i=1,\ldots,p$ and $\ell=1,\ldots,q$, let
$$
M_i = \alpha \varphi_{i,1} \cdots \varphi_{i,k}\qquad
M_{p+\ell} = \gamma_{\ell}\ .
$$
Using the $\Case_n$ term from Proposition \ref{p:case}, set
\begin{eqnarray*}
H &=& \Case_{p+q}M_1\ldots M_p M_{p+1}\ldots M_{p+q}
\\
G &=& \lambda \alpha x_1\ldots x_k \dotsp 
(H\ \rho_1(x_1\ldots x_k)\ldots \rho_{p+q}(x_1\ldots x_k))
\\
F &=& \pad_{K',L'}\ G
\\ 
\theta &=& (\lambda z\dotsp F(zz))\ (\lambda z\dotsp F(zz))
\end{eqnarray*}
The following sequence of reductions is leftmost because,
as long as $\pad_{K',L'}$ is not completely reduced,
there is no $\F$-redex on its right.
$$(R_1)\quad\begin{array}{rll}
\theta\ \corner{a_1}\cdots \corner{a_k}
& \redu_{1,0}&
F\ \theta\ \corner{a_1}\cdots \corner{a_k}
\text{\qquad (cf. Theorem \ref{thm:curry})}
\\
&=&
\pad_{K',L'}\ G\ \theta\ \corner{a_1}\cdots \corner{a_k}
\\
& \redu_{K',L'}& G\ \theta\ \corner{a_1}\cdots \corner{a_k}
\end{array}$$
Let us denote by $A^\sigma$ the term
$A^\sigma_i
= A[\theta/\alpha,\corner{a_1}/x,_1,\ldots, \corner{a_k}/x_k]$.
The leftmost reduction sequence goes on with $\beta$-reductions as follows:
$$(R_2)\ \begin{array}{rll}
G\ \theta\ \corner{a_1}\cdots \corner{a_k}&=&
(\lambda \alpha x_1\ldots x_k \dotsp 
(H\ \rho_1\ldots \rho_{p+q}))\ \theta\ \corner{a_1}\cdots \corner{a_k}
\\
&\redu_{k+1,0}&
H^\sigma\ \rho ^\sigma_1\ldots \rho ^\sigma_{p+q}
\end{array}$$
Now, using Proposition \ref{p:Fterms},
the following leftmost reductions are $\F$-reductions:
 $$
\begin{array}{rcrll}
&&\varphi_{i,j}^\sigma &\redu_{0,L_{\varphi_{i,j}}}&
\corner{f_{i,j}(\vec{a}_{I_{i,j})}}
\\
&&M^\sigma_i &\redu_{0,\sum_{j=1}^{j=k} L_{\varphi_{i,j}}}&
\theta\ \corner{f_{i,1}(\vec{a}_{I_{i,1})}}
\ldots \corner{f_{i,k}(\vec{a}_{I_{i,k})}}
\\
M^\sigma_{p+\ell} &=& \gamma_{p+\ell}^\sigma
&\redu_{0,L_{\gamma_{\ell}}}&
\corner{g_{\ell}(\vec{a}_{J_\ell})}
\\
&&\rho^\sigma_s  &\redu_{0,L_{\rho_s}}&
\corner{r_s(\vec{a}_{I_s)}}
\end{array}
$$
Going with our main leftmost reduction sequence, 
letting
$$
N= (\sum_{i=1}^{i=p} \sum_{j=1}^{j=k} L_{\varphi_{i,j}})
+\sum_{\ell=1}^{\ell=q} L_{\gamma_{\ell}}
+\sum_{s=1}^{s=p+q} L_{\rho_s}
$$
and $s$ be as in condition $(\dagger)_s$ in the statement of the Lemma,
we get
$$(R_3)\quad
\begin{array}{rll}
H^\sigma\ \rho ^\sigma_1\ldots \rho ^\sigma_{p+q}
&=&
\Case_{p+q}M^\sigma_1\ldots M^\sigma_p
M^\sigma_{p+1}\ldots M^\sigma_{p+q}\ \rho ^\sigma_1\ldots \rho ^\sigma_{p+q}
\\
&\redu_{0,N}&
\Case_{p+q}\\
&&\quad(\theta\ \corner{f_{1,1}(\vec{a}_{I_{1,1})}}
                       \ldots \corner{f_{1,k}(\vec{a}_{I_{1,k})}})\\
&& \quad\ldots\\
&& \quad(\theta\ \corner{f_{p,1}(\vec{a}_{I_{p,1})}}
                       \ldots \corner{f_{p,k}(\vec{a}_{I_{p,k})}})\\
&& \quad
     (\corner{g_1(\vec{a}_{J_1})}) \quad \ldots \quad (\corner{g_q(\vec{a}_{J_q})})
\\&& \quad\rho ^\sigma_1 \quad \ldots \quad \rho ^\sigma_{p+q}
\\
&\redu_{3(p+q),0}&
\left\{\begin{array}{ll}
\theta\ \corner{f_{s,1}(\vec{a}_{I_{s,1})}}
                       \ldots \corner{f_{s,k}(\vec{a}_{I_{s,k})}}
&\textit{if }s\leq p
\\
\corner{g_\ell(\vec{a}_{J_\ell})})
&\textit{if }s=p+\ell
\end{array}\right.
\end{array}
$$
Summing up reductions $(R_1)$, $(R_2)$, $(R_3)$, we see that
$$
\begin{array}{rllll}
\theta\ \corner{a_1}\cdots \corner{a_k}
&\redu_{\eta,\zeta}&
\left\{\begin{array}{ll}
\theta\ \corner{f_{s,1}(\vec{a}_{I_{s,1})}}
                       \ldots \corner{f_{s,k}(\vec{a}_{I_{s,k})}}
&\textit{if }s\leq p
\\
\corner{g_\ell(\vec{a}_{J_\ell})})
&\textit{if }s=p+\ell
\end{array}\right.
\end{array}
$$
where\ $\eta = 1+K'+(k+1)+3(p+q)$ and $\zeta = L'+N$.
\\
To conclude, set $\Kmin=k+5+3(p+q)$ and $\Lmin=N$.
If $K\geq \Kmin$ and $L\geq \Lmin$ it suffices to set
$K'=K-(\Kmin-3)$ and $L'=L-\Lmin$ and to observe that
$K'\geq3$ as needed in Lemma \ref{l:padding}.
\end{proof}
%
%
%
%
%
%
\section{ASMs and Lambda Calculus}
\label{s:lambdaASM}
%
All along this section, $\+S = (\+L , P , (\xi, \+J))$
is some fixed ASM (cf. Definition \ref{def:ASM}).
%
%
\subsection{Datatypes and ASM Base Sets}\label{ss:datatypes}
%
The definition of ASM does not put any constraint
on the base sets of the multialgebra.
However, only elements which can be named are of any use,
i.e. elements which are in the range of compositions of (static or dynamic)
functions on the ASM at the successive steps of the run.
\\
The following straightforward result formalizes this observation.
\begin{proposition}\label{p:datatypes_basesets}
Let $(\+L , P , (\xi, \+J))$ be an ASM.
Let $\+U_1,\ldots,\+U_n$ be the base sets interpreting the different sorts
of this ASM.
For $t\in\N$, let
$A^{(t)}_1\subseteq \+U_1$,\ldots, $A^{(t)}_n \subseteq \+U_n$
be the sets of values of all ground good $\F$-terms
(i.e. with no variable) in the $t$-th successor state $\+S_t$ of the initial state
$\+J$ of the ASM.
\medskip\\
1. For any $t\in\N$, $A^{(t)}_1\supseteq A^{(t+1)}_1$, \ldots,
$A^{(t)}_n\supseteq A^{(t+1)}_n$.
\medskip\\
2. $(A^{(t)}_1,\ldots, A^{(t)}_n)$ is a submultialgebra of $\+S_t$,
i.e. it is closed under all static and dynamic functions of the state $\+S_t$.
\end{proposition}
Thus, the program really works only on the elements of the sets
$(A^{(0)}_1,\ldots, A^{(0)}_n)$ of the initial state which
are datatypes defined via mutual inductive definitions using $\xi$ and $\+J$.
%
%
%
%
%
\subsection{Tailoring Lambda Calculus for an ASM}\label{ss:tailor}
%
Let $\F$ be the family of interpretations of all static symbols
in the initial state.
The adequate Lambda calculus to encode the ASM is
$\Lambda_\F$.
\\
Let us argue that this is not an unfair trick.
An algorithm does decompose a task in elementary ones.
But ``elementary" does not mean ``trivial" nor ``atomic",
it just means that we do not detail how they are performed:
they are like oracles.
There is no absolute notion of elementary task.
It depends on what big task is under investigation.
For an algorithm about matrix product, 
multiplication of integers can be seen as elementary.
Thus, algorithms go with oracles.
\\
Exactly the same assumption is done with ASMs:
static and input functions are used for free.
%
%
%
%
%
\subsection{Main Theorem for Type $0$ ASMs}\label{ss:main0}
%

We first consider the case of type $0$ ASMs.
\begin{theorem}\label{thm:main0}
Let $(\+L, P, (\xi,\+J))$ be an ASM with base sets $\+U_1,\ldots,\+U_n$.
Let $A_1$,\ldots, $A_n$ be the datatypes $A^{(0)}_1$,\ldots, $A^{(0)}_n$
(cf. Proposition \ref{p:datatypes_basesets}).
Let $\F$ be the family of interpretations of all static symbols of the ASM
restricted to the datatypes $A_1$,\ldots, $A_n$. 
Suppose all dynamic symbols have arity $0$, i.e. all are constants symbols.
Suppose these dynamic symbols are $\eta_1,\ldots,\eta_k$.
and $\eta_1,\ldots,\eta_\ell$ are the output symbols.
\\
Let us denote by $e_i^t$ the value of the constant $\eta_i$ in the
$t$-th successor state $\+S_t$ of the initial state $\+J$.
\\
There exists $K_0$ such that, for every $K\geq K_0$,
there exists a $\lambda$-term $\theta$ in $\Lambda_\F$
such that,
for all initial values $e_1^0,\ldots,e_k^0$ of the dynamic
constants and for all $t\geq1$,
$$
\begin{array}{rcll}
\theta\ \corner{e_1^0}\ldots \corner{e_k^0}
&\redu_{Kt}&
\theta \corner{e_1^t}\ldots \corner{e_k^t}
&\left\{\text{\begin{tabular}{l}
if the run does not halt\\
nor fail nor clash\\
for steps $\leq t$
\end{tabular}}\right.
\\
\theta\ \corner{e_1^0}\ldots \corner{e_k^0}
&\redu_{Ks}&
\langle\corner{1},\corner{e_1^s}\ldots \corner{e_\ell^s}\rangle
&\text{if the run halts at step $s\leq t$}
\\
\theta\ \corner{e_1^0}\ldots \corner{e_k^0}
&\redu_{Ks}&
\corner{2}
&\text{if the run fails at step $s\leq t$}
\\
\theta\ \corner{e_1^0}\ldots \corner{e_{k}^0}
&\redu_{Ks}&
\corner{3}
&\text{if the run clashes at step $s\leq t$}
\end{array}
$$
Thus, groups of $K$ successive reductions simulate
in a simple way the successive states of the ASM,
 and give the output in due time when it is defined.
\end{theorem} 
\begin{proof}
Use Theorem \ref{thm:normalprogram} to normalize the program $P$.
We stick to the notations of that Theorem.
Since there is no dynamic function, only dynamic constants,
the ASM terms $C_i$ and $I_{i,j}$ name the result of applying to
the dynamic constants a composition of the static functions
(including static constants).
Thus, one can associate good $\F$-terms $\rho_i, \varphi_{i,j}$
to these compositions.
\\
Observe that one can decide if the program halts or fails or clashes
via some composition of functions in $\F$
(use the static equality function which has been assumed,
cf. Definition \ref{def:ASM}).
So enter negative answers to these decisions in the existing
conditions $C_1,\ldots,C_n$.
Also, add three more conditions to deal with the positive answers
to these decisions. These three last conditions are associated to
terms $\gamma_1, \gamma_2, \gamma_3$.
Finally, apply Lemma \ref{l:curryvaria2} (with $p=n$ and $q=3$).
\end{proof}
\begin{remark}
A simple count in the proof of Lemma \ref{l:curryvaria2} allows to bound  
$K_0$ as follows: $K_0=O((\text{size of $P$})^2)$.
\end{remark}

%
\subsection{Main Theorem for All ASMs}\label{ss:main1}
%
Let $\psi$ be a dynamic symbol.
Its initial interpretation $\psi_{\+S_0}$ is given by a composition
of the static objects (cf. Definition \ref{def:ini})
hence it is available in each successor state of the initial state.
In subsequent states $\+S_t$, its interpretation $\psi_{\+S_t}$
is different but remains almost equal to $\psi_{\+S_0}$ : the two
differ only on finitely many tuples. This is so because, at each step,
any dynamic symbol is modified on at most $N$ tuples where $N$
depends on the program.
Let $\Delta\psi$ be a list of all tuples on which $\psi_{\+S_0}$
has been modified.
What can be done with $\psi$ can also be done with
$\psi_{\+S_0}$ and $\Delta\psi$.
Since $\psi_{\+S_0}$ is available in each successor state of the initial state,
we are going to encode $\Delta\psi_{\+S_t}$
rather than $\psi_{\+S_t}$.
Now, $\Delta\psi_{\+S_t}$ is a list and we need to access
in constant time any element of the list. 
And we also need to manage the growth of the list.
\\
This is not possible in constant time with the usual encodings
of datatypes in Lambda calculus.
So the solution is to make $\Lambda_\F$ bigger:
put new constant symbols to represent lists and allow new
$\F$-reduction axioms to get \emph{in one step} the needed information
on lists.
\\
Now, is this fair? We think it is as regards simulation of ASMs.
In ASM theory, one application of the program is done in one unit of time
though it involves a lot of things to do.
In particular, one can get in one unit of time
all needed information about the values of static or dynamic functions
on the tuples named by the ASM program.
What we propose to do with the increase of $\Lambda_\F$ is just
to get more power, as ASMs do on their side.
\begin{definition}
Let $A_1,\ldots,A_n$ be the datatypes involved in functions of $\F$.
If $\varepsilon =(i_1,\ldots,i_m,i)$ is an $(m+1)$-tuple of elements
in $\{1,\ldots,n\}$,
we let $L_\varepsilon $ be the datatype of finite sequences of $(m+1)$-tuples in
$A_{i_1}\times\cdots\times A_{i_m} \times A_i$.
\\
Let $E$ be a family of tuples of elements of $\{1,\ldots,n\}$.
The Lambda calculus $\Lambda_\F^E$
is obtained by adding to $\Lambda_\F$ families of symbols
$$
(F_\varepsilon, B_\varepsilon, V_\varepsilon,
\textit{Add}_\varepsilon, \textit{Del}_\varepsilon)_{\varepsilon\in E}
$$
and the axioms associated to the following intuitions.
For $\varepsilon=(i_1,\ldots,i_m,i)$,
\begin{itemize}
\item[i.]
Symbol $F_\varepsilon$ is to represent the function $L_\varepsilon \to \Bool$
such that, for $\sigma\in L_\varepsilon$,
$F_\varepsilon(\sigma)$ is $\true$ if and only if
$\sigma$ is functional in its first $m$
components.
In other words, $F_\varepsilon$ checks if
any two distinct sequences in $\sigma$ always differ on their first $m$ components.
\item[ii.]
Symbol $B_\varepsilon $ is to represent the function
$L_\varepsilon\times (A_{i_1}\times\cdots\times A_{i_m}) \to \Bool$
such that, for $\sigma\in L_\varepsilon$ and
$\vec{a}\in A_{i_1}\times\cdots\times A_{i_m}$,
$B_\varepsilon(\sigma, \vec{a})$ is $\true$ if and only if
$\vec{a}$ is a prefix of some $(m+1)$-tuple in the finite sequence $\sigma$.
\item[iii.]
Symbol $V_\varepsilon $ is to represent the function
$L_\varepsilon\times (A_{i_1}\times\cdots\times A_{i_m}) \to A_i$
such that, for $\sigma\in L_\varepsilon$ and
$\vec{a}\in A_{i_1}\times\cdots\times A_{i_m}$,
\\- $V_\varepsilon(\sigma, \vec{a})$ is defined if and only if
$F_\varepsilon(\sigma)=\true$ and $B_\varepsilon(\sigma, \vec{a})=\true$,
\\- when defined, $V_\varepsilon(\sigma, \vec{a})$ is the last component
of the unique $(m+1)$-tuple in the finite sequence $\sigma$
which extends the $m$-tuple $\vec{a}$.
\item[iv.]
Symbol $\textit{Add}_\varepsilon $ is to represent the function
$L_\varepsilon \times (A_{i_1}\times\cdots\times A_{i_m} \times A_i)
\to L_\varepsilon$
such that, for $\sigma\in L_\varepsilon$ and
$\vec{a}\in A_{i_1}\times\cdots \times A_{i_m} \times A_i$,
$\textit{Add}_\varepsilon(\sigma, \vec{a})$ is obtained by adding the tuple
$\vec{a}$ as last element in the finite sequence $\sigma$.
\item[v.]
Symbol $\textit{Del}_\varepsilon $ is to represent the function
$L_\varepsilon \times (A_{i_1}\times\cdots\times A_{i_m} \times A_i)
\to L_\varepsilon$
such that, for $\sigma\in L_\varepsilon$ and
$\vec{a}\in A_{i_1}\times\cdots \times A_{i_m} \times A_i$,
$\textit{Del}_\varepsilon(\sigma, \vec{a})$ is obtained by deleting all
occurrences of the tuple $\vec{a}$ in the finite sequence $\sigma$.
\end{itemize}
\end{definition}
Now, we can extend Theorem \ref{thm:main0}.
\begin{theorem}\label{thm:main1}
Let $(\+L, P, (\xi,\+J))$ be an ASM with base sets $\+U_1,\ldots,\+U_n$.
Let $A_1$,\ldots, $A_n$ be the datatypes $A^{(0)}_1$,\ldots, $A^{(0)}_n$
(cf. Proposition \ref{p:datatypes_basesets}).
Let $\F$ be the family of interpretations of all static symbols of the ASM
restricted to the datatypes $A_1$,\ldots, $A_n$. 
Let $\eta_1,\ldots, \eta_k$ be the  dynamic symbols of the ASM.
Suppose $\eta_i$ has type
$\+U_{\tau(i,1)}\times\cdots\times \+U_{\tau(i,p_i)}\to \+U_{q_i}$
for $i=1,\ldots,k$.
\\
Set $E=\{(\tau(i,1),\ldots,\tau(i,p_i),q_i)  \mid i=1,\ldots,k\}$.
\\
The conclusion of Theorem \ref{thm:main0} is still valid
in the Lambda calculus $\Lambda_\F^E$
with the following modification:
\begin{quote}
$e_i^t$ is the list of $p_i+1$-tuples describing
the differences between the interpretations of
${(\eta_i)}_{\+S_0}$ and ${(\eta_i)}_{\+S_t}$.
\end{quote}
\end{theorem} 
%
%
%
%


\end{document}